\documentclass[journal,12pt,onecolumn,draftclsnofoot,]{IEEEtran}
\usepackage{amsmath,amssymb,amsfonts}
\usepackage{algorithm}
\usepackage{algorithmic}
\usepackage{graphicx}
\usepackage{textcomp}

\usepackage{amsthm}
\mathchardef\mhyphen="2D
\newtheorem{theorem}{Theorem}

\newtheorem{lemma}{Lemma}

\theoremstyle{definition}
\newtheorem{definition}{Definition}[]

\usepackage{color}

\usepackage{graphicx,subfigure,verbatim}

\usepackage{enumitem}
\usepackage{multirow}
\usepackage{array}
\usepackage{tabularx}
\usepackage{booktabs}
\usepackage{multirow}
\usepackage{multicol, blindtext}

\ifCLASSOPTIONcompsoc
  \usepackage[nocompress]{cite}
\else
  \usepackage{cite}
\fi
\ifCLASSINFOpdf
\else
\fi
\begin{document}
%
% paper title
% Titles are generally capitalized except for words such as a, an, and, as,
% at, but, by, for, in, nor, of, on, or, the, to and up, which are usually
% not capitalized unless they are the first or last word of the title.
% Linebreaks \\ can be used within to get better formatting as desired.
% Do not put math or special symbols in the title.
\title{Iterative DNA Coding Scheme\\With GC Balance and Run-Length Constraints\\Using a Greedy Algorithm}
%
%
% author names and IEEE memberships
% note positions of commas and nonbreaking spaces ( ~ ) LaTeX will not break
% a structure at a ~ so this keeps an author's name from being broken across
% two lines.
% use \thanks{} to gain access to the first footnote area
% a separate \thanks must be used for each paragraph as LaTeX2e's \thanks
% was not built to handle multiple paragraphs
%
%
%\IEEEcompsocitemizethanks is a special \thanks that produces the bulleted
% lists the Computer Society journals use for "first footnote" author
% affiliations. Use \IEEEcompsocthanksitem which works much like \item
% for each affiliation group. When not in compsoc mode,
% \IEEEcompsocitemizethanks becomes like \thanks and
% \IEEEcompsocthanksitem becomes a line break with idention. This
% facilitates dual compilation, although admittedly the differences in the
% desired content of \author between the different types of papers makes a
% one-size-fits-all approach a daunting prospect. For instance, compsoc 
% journal papers have the author affiliations above the "Manuscript
% received ..."  text while in non-compsoc journals this is reversed. Sigh.

\author{Seong-Joon~Park,
        Yongwoo~Lee,
        and~Jong-Seon~No% <-this % stops a space
\IEEEcompsocitemizethanks{\IEEEcompsocthanksitem S.J Park, Y. Lee, and J. S. No are with the Department of Electrical and Computer Engineering, INMC, Seoul National University, Seoul 08826, Korea.\protect\\

}% <-this % stops an unwanted space
\thanks{}}

\IEEEtitleabstractindextext{%
\begin{abstract}
In this paper, we propose a novel iterative encoding algorithm for DNA storage to satisfy both the GC balance and run-length constraints using a greedy algorithm. DNA strands with run-length more than three and the GC balance ratio far from 50\% are known to be prone to errors. The proposed encoding algorithm stores data at high information density with high flexibility of run-length at most $m$ and GC balance between $0.5\pm\alpha$ for arbitrary $m$ and $\alpha$. More importantly, we propose a novel mapping method to reduce the average bit error compared to the randomly generated mapping method, using a greedy algorithm. The proposed algorithm is implemented through iterative encoding, consisting of three main steps: randomization, M-ary mapping, and verification. The proposed algorithm has an information density of 1.8523 bits/nt in the case of $m=3$ and $\alpha=0.05$. Also, the proposed algorithm is robust to error propagation, since the average bit error caused by the one nt error is 2.3455 bits, which is reduced by $20.5\%$, compared to the randomized mapping.
\end{abstract}

% Note that keywords are not normally used for peerreview papers.
}\par

% make the title area
\maketitle

% To allow for easy dual compilation without having to reenter the
% abstract/keywords data, the \IEEEtitleabstractindextext text will
% not be used in maketitle, but will appear (i.e., to be "transported")
% here as \IEEEdisplaynontitleabstractindextext when the compsoc 
% or transmag modes are not selected <OR> if conference mode is selected 
% - because all conference papers position the abstract like regular
% papers do.
\IEEEdisplaynontitleabstractindextext
% \IEEEdisplaynontitleabstractindextext has no effect when using
% compsoc or transmag under a non-conference mode.

% For peer review papers, you can put extra information on the cover
% page as needed:
% \ifCLASSOPTIONpeerreview
% \begin{center} \bfseries EDICS Category: 3-BBND \end{center}
% \fi
%
% For peerreview papers, this IEEEtran command inserts a page break and
% creates the second title. It will be ignored for other modes.
\IEEEpeerreviewmaketitle

\section{Introduction}\label{sec:introduction}
% Computer Society journal (but not conference!) papers do something unusual
% with the very first section heading (almost always called "Introduction").
% They place it ABOVE the main text! IEEEtran.cls does not automatically do
% this for you, but you can achieve this effect with the provided
% \IEEEraisesectionheading{} command. Note the need to keep any \label that
% is to refer to the section immediately after \section in the above as
% \IEEEraisesectionheading puts \section within a raised box.

% The very first letter is a 2 line initial drop letter followed
% by the rest of the first word in caps (small caps for compsoc).
% 
% form to use if the first word consists of a single letter:
% \IEEEPARstart{A}{demo} file is ....
% 
% form to use if you need the single drop letter followed by
% normal text (unknown if ever used by the IEEE):
% \IEEEPARstart{A}{}demo file is ....
% 
% Some journals put the first two words in caps:
% \IEEEPARstart{T}{his demo} file is ....
% 
% Here we have the typical use of a "T" for an initial drop letter
% and "HIS" in caps to complete the first word.
\IEEEPARstart{R}{ecently}, there is a massive amount of data being produced every day.
%people live in tons of data being created.
In 2025, nearly 175 zettabytes of data are expected to be created \cite{b1}. To handle and store all this information, the need for a new archival storage system has arisen. There are three main important aspects of new archival systems: density, durability, and energy cost. However, current storages such as magnetic tape, hard disk drive (HDD), and solid-state drive (SSD) cannot store exponentially growing data. Therefore, new archival storages that satisfy these constraints and store a huge amount of data have been researched.\par
Among several candidates, deoxyribonucleic acid (DNA) emerges as a suitable medium for a new storage system \cite{b2}, which is called DNA storage. The main idea is to map the data to four nucleotides of DNA, adenine, cytosine, guanine, and thymine, denoted by `A', `C', `G', and `T', respectively. Since the size of nucleotides is extremely small, DNA storage can theoretically store up to one exabyte of data per cubic millimeter. It is already proved by the experiment that in one gram of DNA, nearly 215 petabytes of data can be stored \cite{b3}. Also, data can be stored for more than centuries in DNA storage, and also has low energy cost for storing the data. These are the reasons why DNA storage is more suitable for future archival storage compared to other devices, such as flash memory, HDD, and magnetic tape \cite{b4}. Because of these reasons, DNA storage is currently an active area of research in storage systems.\par

%DNA fountain journal
% J. Bornholt, R. Lopez, D. M. Carmean, L. Ceze, G. Seelig and K. Strauss, "Toward a DNA-based archival storage system," IEEE Micro, vol. 37, no. 3, pp. 98-104, June 2017.

Despite these many advantages, DNA storage has several shortcomings that should be overcome. DNA storage has a relatively high synthesis cost. This problem requires an efficient encoding algorithm to store a large amount of data in less number of nucleotides. Also, the error rate of DNA storage is greatly influenced by the biochemical structure of DNA. The following two biochemical constraints should be met because these can cause a high error rate in both the synthesizing and sequencing processes.\par
\begin{itemize}[leftmargin=20pt]
  \item \textit{GC balance ratio}: GC balance ratio is defined by the ratio between the number of G and C nucleotides and the number of A and T nucleotides. This ratio needs to be near 0.5 because high or low GC balance ratio causes high drop out rates and polymerase chain reaction (PCR) errors \cite{b3}, \cite{b5}, \cite{b6}. Balancing this ratio leads to a lower error rate during both synthesis and sequencing processes. Therefore, it is important to balance this ratio.\\
  \item \textit{Maximum run-length}: Maximum run-length is the maximum number of consecutive identical nucleotides in the DNA strand. It is known that substitution and deletion error rates increase if the maximum run-length is longer than six \cite{b5}-\cite{b6}. This would also cause a high error rate during the DNA storage processes.
   
\end{itemize}
Therefore, these two biochemical constraints should be met to have better performance in DNA storage.
\par

There are many studies to preserve these two constraints. Goldman \textit{et al.} \cite{b7} compressed the raw data using Huffman coding and preserved the maximum run-length limit. However, they did not preserve the GC balance ratio. Xue \textit{et al.} \cite{b8} did not preserve the maximum run-length limit, but made GC balance ratio exactly 0.5 with deletion/insertion/mutation error correction. Some other researches preserved both the GC balance ratio and the maximum run-length limit. Erlich and Zielinski \cite{b3} used Luby transform and screening method to preserve these two constraints and proposed DNA Fountain code which stores the data with high physical density. Lu \textit{et al.} \cite{b9} also applied DNA Fountain encoding scheme to preserve both constraints and proposed a new log-likelihood ratio calculation for low-density parity-check codes. Yazdi \textit{et al.} \cite{b10} preserved the GC balance ratio by partitioning data into eight bases and using specialized constrained coding. Also, they limited the maximum run-length by using homopolymer check codes. Mishra \textit{et al.} \cite{b11} used the minimum Huffman variance tree and its complementary tree to compress the data and limit run-length at most 1. Also, they limited the GC balance ratio to nearly 0.5. Immink and Cai \cite{b12}, \cite{b20} proposed mathematical DNA coding that preserves both the GC balance ratio and maximum run-length limit. Wang \textit{et al.} \cite{b21} proposed DNA coding with high information density of 1.92 bit/nt, which converts 23 bits to 11 nt DNA sequence. They preserved both the GC balance ratio between $40\%$ to $60\%$ and the maximum run-length at most three by using a finite state transition diagram. Nguyen \textit{et al.} \cite{b22} proposed constrained coding with information density of $1.92$ bit/nt in error-free channel. They used run-length limit code to preserve the maximum run-length at most three and Knuth's balancing technique to preserve the GC balance ratio between $40\%$ to $60\%$. Also, they applied error-correcting code to handle a single substitution/insertion/deletion error.\par
In this paper, we propose a new iterative DNA coding algorithm that satisfies two constraints mentioned above with high information density. More importantly, to reduce the average bit error compared to the randomly generated mapping method, we proposed a novel mapping method of mapping one 48-ary symbol to the three nt DNA sequence. The average bit error caused by one nt error is $2.3455$ bits, which is reduced by $20.5\%$, compared to the randomized mapping method. Therefore, the proposed algorithm is robust to error propagation. Also, this algorithm guarantees high flexibility for various constraints. The DNA strand with the GC balance ratio between $0.5\pm \alpha$ and maximum run-length of $m$ can be obtained. One can flexibly set the desired values $\alpha$ and $m$ for the proposed iterative encoding algorithm. Not only these two constraints but also other desired constraints such as avoiding specific patterns (for example, primers) of DNA can be set. 
%modi: four diff. steps -> four steps
This is implemented through iterative encoding, which consists of three steps: randomization, $M$-ary mapping, and verification. Since the GC balance ratio and maximum run-length limit are two main constraints for DNA storage, this paper mostly focuses on preserving these two constraints, especially for the typical DNA data storage of $m=3$ and $\alpha=0.05$ \cite{b3}, \cite{b6}.\par
This paper is organized as follows. {Section \ref{section2}} contains definitions and preliminaries to understand concepts in DNA storage. In {Section \ref{section3}}, we describe the proposed iterative coding algorithm consisting of three steps that compress the raw data with two constraints. In {Section \ref{section4}}, we calculate the information density by using well-known raw data files and compare them with existing works. {Section \ref{section5}} concludes the paper.

\vspace{10pt}
\section{Definitions and Preliminaries}
\label{section2}
%modi: (0,0), (0,1)이 필요한가??
In this paper, we define four types of nucleotides as quaternary numbers: $A=0$, $C=1$, $G=2$, and $T=3$.
%\subsection{GC Balance Ratio}
\begin{definition}
Let $\mathbf{v}$ $=(v_1, ..., v_n)$, $v_i\in\{0,1,2,3\}$, denote a DNA strand of length $n$ with quaternary elements. Define the number of `1' and `2' in $\mathbf{v}$ as $\eta(G,C)$. Then GC balance ratio is defined as

\[r_{GC} = \frac{\eta(G,C)}{n}.\]

In this paper, we say that DNA strand is balanced if $0.5-\alpha \le r_{GC} \le 0.5+\alpha$ for small $\alpha$.
\end{definition}
\par
%\subsection{Maximum Run Length}
\begin{definition}
Let $m$ be the maximum number of consecutively repeated identical nucleotides in the DNA strand, and we call it maximum run-length. In this paper, we say that the DNA strand satisfies the maximum run-length limit $m$.
\end{definition}

\vspace{10pt}
\section{Proposed Iterative Encoding Algorithm}
\label{section3}
%modi: spell check!
% four algorithm라고 하면 서로 다른 네개의 메서드가 되는데 사실 그건아니까.
In this section, we propose a new encoding algorithm of DNA coding with two constraints in detail. It is implemented through iterative encoding, which consists of three steps: randomization, $M$-ary mapping, and verification. DNA sequences created by the proposed encoding algorithm satisfy given constraints, the GC balance ratio and the maximum run-length. It is worth noting that, depending on the application, any efficient source coding can be applied before the proposed encoding algorithm. Since raw data files are used in our experiment, we apply the source coding step before the proposed encoding algorithm. Here, we perform binary source coding, considering each divided block of length $k$ as a source symbol. Then the compressed output $F_{comp}$, as a binary form, is obtained.

% needed in second column of first page if using \IEEEpubid
%\IEEEpubidadjcol
\begin{table}[t!]
\centering
\scriptsize
\renewcommand{\arraystretch}{1.2}
\caption{Comparison of substitution errors across different bases and average bit error}
\label{table_error}
\resizebox{\textwidth}{!}{
\begin{tabular}{|c|c|c|c|c|c|c|c|c|c|c|c|c|}
\hline
Base to base       & G to A & G to T & C to A & C to T & T to C & A to G & T to A & A to T & T to G & G to C & A to C & C to G \\ \hline
Sub. error prob.$(\%)$& 14.133 & 13.773 & 8.894  & 7.842  & 7.142  & 7.067  & 7.050  & 7.046  & 6.948  & 6.889  & 6.826  & 6.387  \\ \hline
Average bit error & 2      & 2.357  & 2.214  & 2.357  & 2.357  & 2      & 2.5    & 2.5    & 2.357  & 2.857  & 2.214  & 2.857  \\ \hline
\end{tabular}
}
\end{table}

\subsection{Randomization}
\label{randomization}
After obtaining the compressed output $F_{comp}$, we perform randomization to make the GC balance ratio between $0.5\pm \alpha$. In the coding theory, there is a scheme called guided scrambling \cite{b19}, which is similar to the randomization step. It is applied with the verification step to satisfy the GC balance ratio. In the randomization process, the input value $r$ is used to generate a randomized sequence using a random number generator function such as a hash function to obtain the output $h(r)$. After that, we perform bitwise XOR operation of $F_{comp}$ and $h(r)$ to obtain $F_{rand}$, which is randomized binary data of $F_{comp}$. Generally, the length of $F_{comp}$ is longer than $h(r)$, and thus $F_{comp}$ is partitioned into the block of the same length as $h(r)$. Since $h(r)$ is randomized sequence, it also guarantees that $F_{rand}$ is randomized. Also, the randomization step is performed iteratively, and it is explained more specifically in Section \ref{section_verification}.\par

\subsection{$M$-ary Mapping}
\label{mapping}
The next step is to map the binary sequence to the DNA sequence. In this step, the binary sequence is converted to $M$-ary symbols ($M=3\cdot4^{m-1}$) and mapped to the DNA sequence by using the mapping table to satisfy the maximum run-length $m$. 
The most important feature of this step is that we propose the mapping table to reduce the average bit error when one nt error in the DNA sequence occurs. The mapping table is greedily constructed according to the substitution error probability between bases based on our experiment. In this subsection, we present two different mapping methods. While mapping the binary sequence to the DNA sequence, there is a trade-off between an information density and an error propagation. The first method is to have an advantage on an information density, and the second method is to have an advantage on an error propagation.
The first method is to convert the whole binary sequence to $M$-ary symbols. The second method is to partition the binary sequence into blocks of 11 bits and convert it to two 48-ary symbols in the case of $m=3$. The first method has  higher information density compared to the second method. However, the second method is robust to error propagation.

\subsubsection{The Mapping Method Using a Greedy Algorithm}
\label{section_greedy}
In both methods, the binary sequence is converted to $M$-ary symbols and mapped to the DNA sequence using a mapping table. To ensure that the maximum run-length $m$ is satisfied after the mapping step, we should have $M=4^{m-1}\cdot3$. Therefore, in the typical DNA data storage of $m=3$, $48$-ary mapping table is required in the mapping step.

\begin{lemma}
\label{lemma_run_length}
Let $\mathbf{x}$ be the vector of length $m$ and $\mathbf{X}_{i}$ be the set of vectors defined as
% defined by
\[\mathbf{X}_{i} = \{\mathbf{x}=(x_1,...,x_m)\in \{0,1,2,3\}^m\mid x_i \neq x_{i+1}\},\]
for an integer $i$.\par
If any vectors in $\mathbf{X}_{i}$ are appended together, they have run-length less than or equal to $m$.
For example, $\mathbf{X}_{2}=\{(0,0,1),(0,0,2), ... (3,3,1), (3,3,2)\},$ where $m=3$ and $i=2$ and then, $x_2 \neq x_3$ for all $\mathbf{x} =(x_1,x_2,x_3) \in \mathbf{X}_{2}$.
\end{lemma}

\begin{proof}
 Let $\mathbf{u},\mathbf{v}\in \mathbf{X}_{i}$, that is, $u_i\neq u_{i+1}$ and $v_i\neq v_{i+1}$.
Then there are two cases to be considered.

\begin{enumerate}[label=(\roman*),leftmargin=20pt]
    \item When $u_{i+1} \neq v_i$, the maximum run-length is $m-1$ when $\mathbf{u}$ and $\mathbf{v}$ are appended together as $(\mathbf{u}\mid\mathbf{v})$.
    \item When $u_{i+1} = v_i$, the maximum run-length is $m$ when $\mathbf{u}$ and $\mathbf{v}$ are appended together as $(\mathbf{u}\mid\mathbf{v})$.
\end{enumerate}

For both cases, appending vectors in different way such as   $(\mathbf{v}\mid\mathbf{u})$, also has the same result. Therefore, vectors formed by appending any vectors in $\mathbf{X}_{i}$ have run-length less than or equal to $m$.
\end{proof}

\begin{theorem}
% modi: 상호참조 예시.
\label{thm_run_length}
Let $\mathbf{V}$ be the set of vectors consist of $\mathbf{v}$ $=(v_1, ..., v_m),\,v_i\in \{0, 1, 2, 3\}$. In $\mathbf{V}$, there exist at least $3\cdot 4^{m-1}$ vectors, whose combination of any of these elements have run-length at most $m$.
\end{theorem}

\begin{proof}
 It is easy to check that in Lemma \ref{lemma_run_length}, the vector set $\mathbf{X}_{i}$ has $3\cdot 4^{m-1}$ different vectors. Vectors formed by appending any vectors in $\mathbf{X}_{i}$ have the run-length at most $m$. Therefore, this ensures that in vector set $\mathbf{V}$, there are at least $3\cdot 4^{m-1}$ of vectors, in which combination of any of these vectors have maximum run-length $m$.
\end{proof}

\begin{table}[t!]
\centering
\small
\renewcommand{\arraystretch}{1.2}
\caption{Gray code of 48-ary symbols}
\label{table_gray}
\begin{tabular}{|c|c|c|c|c|c|}
\hline
Binary & Symbol & Binary & Symbol & Binary & Symbol \\ \hline
000000 & 0      & 011011 & 27     & 101000 & 40     \\ \hline
000001 & 1      & 011010 & 26     & 101010 & 42     \\ \hline
000011 & 3      & 011110 & 30     & 101011 & 43     \\ \hline
000010 & 2      & 011111 & 31     & 101001 & 41     \\ \hline
000110 & 6      & 011101 & 29     & 101101 & 45     \\ \hline
000111 & 7      & 011100 & 28     & 101111 & 47     \\ \hline
000101 & 5      & 010100 & 20     & 101110 & 46     \\ \hline
000100 & 4      & 010101 & 21     & 101100 & 44     \\ \hline
001100 & 12     & 010111 & 23     & 100100 & 36     \\ \hline
001101 & 13     & 010110 & 22     & 100101 & 37     \\ \hline
001111 & 15     & 010010 & 18     & 100111 & 39     \\ \hline
001110 & 14     & 010011 & 19     & 100110 & 38     \\ \hline
001010 & 10     & 010001 & 17     & 100010 & 34     \\ \hline
001011 & 11     & 010000 & 16     & 100011 & 35     \\ \hline
001001 & 9      & 011000 & 24     & 100001 & 33     \\ \hline
011001 & 25     & 001000 & 8      & 100000 & 32     \\ \hline
\end{tabular}
\end{table}

According to {Theorem \ref{thm_run_length}}, there are $3\cdot 4^{m-1}$ different vectors in $\mathbf{V}$, and each vector can be mapped to $3\cdot 4^{m-1}$-ary symbols. Let $M = 3\cdot 4^{m-1}$. In the typical DNA storage of $m=3$, we should have $M=48$ and the binary sequence is converted to 48-ary symbol. After that, each 48-ary symbol is mapped to three nt DNA sequence using 48-ary mapping table. Rather than forming the randomized mapping table, we propose the mapping table to reduce the average bit error when one nt DNA sequence error occurs. Here, the mapping table is greedily constructed according to the substitution error probability between bases. We have an experimental result and obtain the substitution error probability between bases \cite{b15}. A more detailed explanation for the experiment is mentioned in Section \ref{section_simulation_mapping}. Table \ref{table_error} shows the comparison of substitution errors across different bases according to our experiment. A substitution error from G to A has the highest probability, and a substitution error from C to G has the lowest probability.\par
We greedily form the mapping table according to Table \ref{table_error}. Starting from (A, A, C), we find the sequence that is most likely to be changed when one nt error occurs among 47 candidates, excluding (A, A, C). Table \ref{table_error} is used to find the next DNA sequence. When the DNA sequence, which is most likely to be changed when one nt error occurs, is chosen, we find the next DNA sequence from the chosen DNA sequence, eliminating the chosen DNA sequence from candidates. In this way, we find the next DNA sequence greedily until all 48 DNA sequences are used. For example, (A, A, C) is most likely to change to (A, A, A) since the substitution error from C to A has the highest probability. However, (A, A, A) is not one of the candidates of 47 sequences since the last two bases are the same. Therefore, the next highest probability is the substitution from C to T, and it is one of the candidates of 47 sequences. Therefore, the next DNA sequence would be (A, A, T). If there is no remaining candidate that differs only one nt, the DNA sequence with two nt differences with the highest probability is chosen for the next DNA sequence.\par

\begin{table}[t!]
\centering
\renewcommand{\arraystretch}{1.2}
\small
\caption{Example of 48-ary mapping table for $m=3$ using a greedy algorithm \label{table_map}}
\begin{tabular}{|c|c|c|c|c|c|}
\hline
DNA & Symbol & DNA & Symbol & DNA & Symbol \\ \hline
AAC & 0      & TCT & 27     & GTG & 40     \\ \hline
AAT & 1      & CCT & 26     & ATG & 42     \\ \hline
GAT & 3      & CCA & 30     & TTG & 43     \\ \hline
TAT & 2      & CCG & 31     & CTG & 41     \\ \hline
TGT & 6      & CAG & 29     & CTA & 45     \\ \hline
CGT & 7      & CAT & 28     & CGA & 47     \\ \hline
AGT & 5      & CAC & 20     & AGA & 46     \\ \hline
AGC & 4      & TAC & 21     & GGA & 44     \\ \hline
ATC & 12     & TGC & 23     & TGA & 36     \\ \hline
ATA & 13     & TTC & 22     & CGC & 37     \\ \hline
GTA & 15     & TTA & 18     & CTC & 39     \\ \hline
GCA & 14     & TCA & 19     & GTC & 38     \\ \hline
ACA & 10     & TCG & 17     & GAC & 34     \\ \hline
ACG & 11     & TAG & 16     & GGC & 35     \\ \hline
ACT & 9      & AAG & 24     & GGT & 33     \\ \hline
GCT & 25     & GAG & 8      & GCG & 32     \\ \hline
\end{tabular}
\end{table}

Now all 48 different DNA sequences should be mapped to all 48-ary symbols, respectively. Here, 48-ary symbols, from 0 to 47, can be expressed in 6 bits, from 000000 to 101111. We can reduce the average bit error by constructing the mapping table using a greedy algorithm. In this step, the mapping table can be constructed to have only one bit difference for adjacent symbols, such as Gray code. Table \ref{table_gray} is the Gray code of 48-ary symbols in the binary form, and adjacent symbols have only one bit difference. By mapping DNA sequences to 48-ary symbols in the order of Table \ref{table_gray}, less number of bit errors would occur when one nt error occurs in the DNA sequence. Table \ref{table_map} is the example of 48-ary mapping table for $m=3$. The greedy algorithm is applied in DNA sequences and DNA sequences are mapped with Gray code of 48-ary symbols. This mapping method would decrease the average bit error by $20.5\%$, compared to the one in the randomly generated mapping table.The more detailed explanation is given in Section \ref{section_simulation_mapping}.\par
%To show the effectiveness of the proposed mapping method, we found the average bit errors for each substitution error probability between bases. The last row of Table \ref{table_error} shows the average bit errors for each case of substitution errors when one nt error occurs in the DNA sequence. Therefore, the average bit errors when one nt error occurs are 2.3455 bits. However, in the randomly generated mapping table, every 48-ary symbol could be changed to any other 48-ary symbol when one nt error occurs. In this case, the average bit errors are 2.9504 bits. In other words, the proposed mapping method could reduce the bit error of 0.6049 bits, which is a $20.5\%$ reduction, compared to the randomly generating mapping method. Therefore, in this step, the DNA sequence not only satisfies the maximum run-length $m$, but also reduce the bit error $20.5\%$.

% modi : theorem, table, figure의 경우, 숫자를 각자 쓰지 말고 \label과 \ref를 이용해서 상호참조할것.
% 기울임체는 어떤 term을 새롭게 정의할 때 딱 한번만 쓰는것

%% ylee
%% 두괄식으로: binary sequence 통채로 48 ary로 바꾸면 gc ratio가 맞다, 그러나, 작은 bit단위로 잘라서 하는게 효율적이다.
%% 이렇게 작은 단위로 잘라서 할 경우에도 GC ratio가 맞도록 하는 방법을 설명할 것이다.
%% m=3말고도 generally applicable하다고 써야할듯
\subsubsection{The Mapping Method With High Information Density}
\label{section_ideal}
The first mapping method converts the whole binary sequence to $M$-ary symbols and maps $M-ary$ symbols to DNA sequences using a mapping table in Table \ref{table_map}. Since the binary sequence is randomized after the randomization step, the occurrences of all 48-ary symbols are equally likely. Therefore, it is easily known that the GC balance ratio is satisfied in this case. When the whole randomized binary data is mapped to $M$-ary symbols, the length of the binary sequence is shortened by the ratio of $\frac{1}{\log_2{M}}$ times. Then, $M$-ary symbols can be mapped to vectors, each representing $m$ nt long DNA sequence . In conclusion, $F_{rand}$ shortens $\frac{m}{\log_2{M}}$ times. Therefore, the information density of this step would be

\begin{equation}
\label{equ_1}
    \frac{\log_2{(3\cdot 4^{m-1})}}{m}= 2- \frac{1}{m}(2-\log_2{3}).
\end{equation}

According to (\ref{equ_1}), as $m$ becomes larger, the information density converges to the ideal upper bound, which is 2 bits/nt. For $m=3$ and $m=4$, the information density of the mapping step is 1.8617 bits/nt and 1.8962 bits/nt, respectively. Since, the capacity for $m=3$ is $1.9824$ bits/nt \cite{b22}, the coding efficiency, which is defined as dividing the information density by the capacity, is $94\%$. When the length of the binary sequence is $N$, the mapping method with high information density has complexity $O(N^2 \log N)$ , which is complexity of a base conversion.\par
\subsubsection{The Mapping Method With Robustness to Error Propagation}
\label{section_practical}
The second mapping method partitions the binary sequence into blocks and converts blocks to 48-ary symbols. In the first mapping method, converting whole binary sequence to $M$-ary symbols has several shortcomings. First, it has complexity $O(N^2\log N)$, where $N$ is the length of the sequence. However, in the second mapping method, it has complexity $O(N)$. Second, it has high error propagation. Because of its high error propagation, when one nt error occurs during the DNA storage process, the whole binary sequence would be corrupted. However, when the binary sequence is partitioned into blocks, only the block with error is corrupted. Therefore, it would be efficient to implement the second method in the DNA storage channel because of its low complexity and robustness to error propagation. Therefore, the binary sequence is partitioned into blocks of 11 bits and converted to decimal and two-digit 48-ary symbols in this method. Since the binary input data is randomized in the randomization step, the occurrence from 0 to 2047 for 11 bits is equally likely. However, when 11 bits are converted to two 48-ary symbols from (0,0) to (42,31), each 48-ary symbol does not occur equally likely. However, when the mapping step is combined with the verification and iterative encoding, which is the next step of the proposed encoding algorithm, it is ensured that the GC balance ratio is satisfied in the second method. A more detailed explanation and the proof is given in Section \ref{section_verification}. In the second method, 48-ary symbols are mapped to DNA sequences according to Table \ref{table_map}, the same as in the first method. However, the information density of this case is $\frac{11}{6}=1.8333$ bits/nt, which is slightly lower than the first method.

% Please add the following required packages to your document preamble:
% \usepackage{booktabs}
\begin{table}[t!]
\centering
\small
\renewcommand{\arraystretch}{1.2}
\caption{The minimum value of $\alpha$ for $I=4, 8$ and various $n$ with $\epsilon=10^{-4}$\label{table_randomization}}
\begin{tabular}{|c|c|c|c|c|c|}
\hline
\multirow{2}{*}{$I$} & \multicolumn{5}{c|}{$n$ (nt)}        \\ \cline{2-6} 
                   & 100  & 150  & 200  & 250   & 300   \\ \hline
4                  & 0.07 & 0.06 & 0.05 & 0.048 & 0.044 \\ \hline
8                  & 0.04 & 0.04 & 0.03 & 0.028 & 0.027 \\ \hline
\end{tabular}
\end{table}
\subsection{Verification and Iterative Encoding}
\label{section_verification}
After the mapping step, we partition the data into $n$-nt long DNA sequences and check whether the constraint is satisfied or not. Since the cost of synthesizing a long DNA sequence is very high, a DNA sequence should be partitioned into shorter sequences, about 150 nt to 300 nt, in the DNA storage \cite{b3}, \cite{b22}. We call the last step of the proposed coding algorithm, \textit{verification}: to find DNA sequences until the desired constraints are satisfied. This method is not defined in the field of DNA storage, and thus we define the term \textit{verification} as follows.

\begin{definition}
For a given DNA sequence and set of constraints that the DNA sequence should meet, the verification is defined by the process of checking whether all constraints are satisfied.
\end{definition}

\begin{algorithm}[h!]
\caption{Iterative Encoding \label{alg_verification}}
\textbf{Input:} input data $X$, desired constraint $U$\\
\textbf{Output:} valid DNA sequence $F_{dna}$.\\
\textbf{Initialization:} $r\leftarrow 0$
\begin{algorithmic}[1]
\STATE $F_{comp}\leftarrow \textsc{SourceCoding}(X)$\\
\WHILE{true}
    \STATE $F_{rand}\leftarrow \textsc{Randomization}(F_{comp}, r)$\\
    \STATE $F_{mapping}\leftarrow \textsc{M-aryMapping}(F_{rand})$\\
    \STATE $F_{partitioned}\leftarrow \textsc{Partitioning}(F_{mapping})$\\
    \IF{$\textsc{Verification}\left(U, (F_{partitioned}|r)\right)$ succeeds}
        \RETURN $F_{dna}\leftarrow (F_{partitioned}|r)$
    \ENDIF
    \STATE $r \leftarrow r+1$
\ENDWHILE
\end{algorithmic}
\end{algorithm}

The iterative encoding in Algorithm \ref{alg_verification} is an algorithm to obtain the DNA sequence which satisfies the desired constraints. When the constraints are not satisfied, we go back to the randomization step and encode the data again until the constraints are satisfied.\par
As we mentioned in Section \ref{section_ideal}, the binary sequence is randomized after the randomization step. Thus the occurrences of each 48-ary symbols are equally likely. Also, since the occurrences of each four bases (A, C, G, T) are the same as 36, it is easily known that the GC balance ratio is satisfied in the first case. However, since 11 bits are converted to the six nt DNA sequence in the mapping method with robustness to error propagation, the iterative encoding needs to be applied to satisfy the GC balance ratio. In the following theorem, we derive the lower bound of the required iteration number that satisfies the GC balance ratio between $0.5\pm \alpha$ in the case of $m=3$ after the 48-ary mapping step. In the case of $m=3$, let $\mathbf{v}=(v_1, ..., v_n)$, $v_i\in\{0,1,2,3\}$, denote a DNA sequence of length $n$ with quaternary elements. Let $p_{GC}$ be the probability of vector $\mathbf{v}$ to satisfy the GC balance ratio between $0.5\pm \alpha$ within the required number of iterations, $I$, defined by
\begin{equation}
\label{equ_pgc}
    p_{GC} = 1-\{1-P(0.5-\alpha \le r_{GC} \le 0.5+\alpha)\}^{I} \ge 1-\epsilon,
\end{equation}
where $r_{GC}$ denotes the GC balance ratio of $\mathbf{v}$. Let $\mathbf{X}_q$ be the random variable representing the number of occurrences of G or C in the $q$th six tuple $(v_{6q+1}, v_{6q+2}, v_{6q+3}, v_{6q+4}, v_{6q+5}, v_{6q+6})$, for $0\le q \le \frac{n}{6}-1$. We assume that $\mathbf{X}_q$'s are statistically independent. Let $p_{l}$ be the probability of $\mathbf{X}_q=l$, $0 \le l \le 6$. Then, we have 
\begin{equation}
\label{equ_aj}
    (p_{0}+p_{1}x+p_{2}x^2+...+p_{6}x^6)^{\frac{n}{6}}=\sum_{j=0}^{n}a_jx^j,
\end{equation}
where $a_j$ denotes the probability of the number of occurrences of G or C being $j$ and let
\begin{equation}
\label{equ_p}
    p(\alpha,n)=\sum_{(0.5-\alpha)n \le j \le (0.5+\alpha)n} a_j,
\end{equation}
which means $P(0.5-\alpha \le r_{GC} \le 0.5+\alpha)$ in (\ref{equ_pgc}).

\begin{theorem}

\label{thm_randomization}

Let the GC balance ratio between $0.5 \pm \alpha$ is satisfied in the proposed iterative encoding algorithm, for $m=3$. Then, for $n=200$ and $I=4$, $\alpha\ge0.05$ is achieved.

\end{theorem}

\begin{proof}
Let $Y$ be a random variable representing the number of the base $G$ and $C$ in a vector $\mathbf{v}$. Since every 11 bits are converted to six-nt DNA sequence, $\mathbf{X}_q$, for $0 \le q \le \frac{n}{6}-1$, are independent and each $\mathbf{X}_q$ has the same $p_0, \dots, p_6$ in \eqref{equ_aj}. Therefore, the probability of $Y=j$ in $\mathbf{v}$ can be expressed as $a_j$ as in (\ref{equ_aj}). In other words, $a_j=P[Y=j]$. To satisfy balancing constraint for $I=4$ iterations, (\ref{equ_pgc}) can be rewritten as 
%modi: \approx 1 지움.
\[p_{GC}  = 1-\{1-p(\alpha,n)\}^{4}\ge 1-\epsilon.\]
Therefore, the lower bound of $p(\alpha,n)$ is
\[p(\alpha,n)\ge 1-\epsilon^{\frac{1}{4}}.\]
Since for $\epsilon=10^{-4}$, we say the probability $p_{GC}$ approximates to 1 and for $n=200$ and $I=4$,
\begin{equation}
\label{equ_alpha}
\sum_{(0.5-\alpha)n \le j \le (0.5+\alpha)n} a_j\ge 1-\epsilon^{\frac{1}{4}}.    
\end{equation}

Therefore, according to \ref{equ_alpha}, for $\alpha\ge0.05$ can be achieved for $n=200$ and $I=4$. 
\end{proof}

In the typical DNA storage with $m=3$, when 11 bits are converted to six nt DNA sequence using 48-ary mapping in Table \ref{table_map}, we can express (\ref{equ_aj}) as
\[(\frac{16}{2048}+\frac{148}{2048}x+\frac{487}{2048}x^2+\frac{724}{2048}x^3+\frac{505}{2048}x^4\]
\[+\frac{152}{2048}x^5+\frac{16}{2048}x^6)^{\frac{n}{6}}=\sum_{j=0}^{n}a_jx^j,\]
where $a_j$ and $p(\alpha,n)$ can be obtained. Table \ref{table_randomization} shows the minimum value of $\alpha$  for $I=4, 8$ and various $n$, using Theorem \ref{thm_randomization} for $\epsilon=10^{-4}$. For example, four times of iteration satisfies the GC balance ratio between $[0.45, 0.55]$ for $n=200$, and $[0.456, 0.544]$ for $n=3$. For the longer length and larger iteration number, the tighter GC balance ratio can be satisfied. Therefore, for the longer DNA sequences, the proposed iterative encoding algorithm can satisfy the tighter GC balance ratio for the fixed iteration number of $I=4$. In addition, since the maximum DNA synthesis length becomes longer recently, the tighter GC balance ratio can be achieved in the proposed iterative encoding algorithm.\par

We chose $I=4$ in Theorem \ref{thm_randomization} due to following reasons: it is good to use a multiple-of-four $I$ because $I$ should be converted to quaternary DNA symbols, and in practice, $I=4$ is enough to satisfy the GC balance ratio in the proposed iterative encoding algorithm. Since randomization input value $r$ must be stored to obtain randomized output $h(r)$ in the decoding procedure, the iteration number at most four can be expressed as one quaternary symbol. Therefore, a multiple-of-four $I$ is adequate for the iteration number, but only $I=4$ is enough to satisfy the GC balance ratio in the proposed iterative encoding algorithm. Starting from $r=0$, as iteration number increases, $r$ increases together until the GC balance ratio is satisfied. In other words, $r=I+1$. For example, $I=3$ means $r=2$, and the input value $r$ from 0 to 3 can be expressed as one quaternary symbol: $0\rightarrow A$, $1\rightarrow C$, $2\rightarrow G$, $3\rightarrow T$. Therefore, only one redundancy is appended on the front of the partitioned DNA sequence of $F_{mapping}$, which is $F_{partitioned}$, for satisfying the desired constraints. If the constraint is satisfied after appended the input value $r$, we obtain the final DNA sequence $F_{dna}$. According to Theorem \ref{thm_randomization} and Table \ref{table_randomization} the GC balance ratio between $0.5\pm 0.05$ can be satisfied for $n=200$ within 4 times of iterations. Also, since only one-nt long DNA sequence is required for $r$, this occupies a very small part compared to the data.\par

Another important feature of the proposed encoding algorithm is flexibility. The desired constraints could be not only GC balance ratio but also other specific constraints. For example, avoiding particular patterns in the DNA strand could be possible. In future DNA storage, one pool might contain multiple files. Also, each file needs a primer, which is a pattern of DNA that is used for the polymerase chain reaction (PCR) or sequencing process, and different files have a different pattern of primer \cite{b18}.
To not be confused with primer and payload, each file should avoid the primers of other files' patterns in its payload. We can set undesired patterns and encode them iteratively using verification until that pattern does not appear in the payload.\par
In addition to flexibility, the proposed algorithm is very robust to errors. As we mentioned in Section \ref{section_greedy}, we greedily constructed the mapping table, using the error probability obtained from the DNA storage experiment to reduce the average bit error when one nt DNA sequence error occurs. Therefore, the average bit error for the proposed encoding algorithm is reduced by $20.5\%$ compared to the randomly generated mapping table. Also, as mentioned in Section \ref{section_practical}, since the binary sequence is converted to the DNA sequence by partitioning it into blocks, the DNA sequence is robust to error propagation. These two features make the proposed encoding algorithm robust to errors. It is very important because, like other channels, the DNA storage channel also has several types of errors, such as substitution, insertion, and deletion errors. Also, there are many works that apply error-correcting code in the DNA storage \cite{b3}, \cite{b8}, \cite{b10}, \cite{b15}, \cite{b18}, \cite{b22}. Therefore, when an error-correcting code is applied in the proposed encoding algorithm, it would be very efficient to apply the proposed encoding algorithm to the DNA storage.\par

\subsection{Decoding Algorithm}
\label{section_decoding}
The decoding algorithm is performed through the inverse operation of the encoding algorithm. The decoding of the DNA strand is done by the following step:

\begin{enumerate}[leftmargin=20pt]
    \item Separate the value $r$ from $F_{dna}$.
    \item Map $F_{mapping}$ to $M$-ary symbols and obtain $F_{M \mhyphen ary}$.
    \item Obtain $F_{rand}$ by converting $F_{M \mhyphen ary}$ into binary sequence.
    \item XOR $h(r)$ with $F_{rand}$ to obtain $F_{comp}$.
    \item Obtain the raw input data $F$ by source decoding if source coding is applied.
    
\end{enumerate}

By using this algorithm, $F_{dna}$ can be uniquely decoded, and the raw data input file $F$ would be obtained.

\begin{table}[t!]
\centering
\renewcommand{\arraystretch}{1.2}
\small
\caption{Comparison of information density (bits/nt) with source coding\label{table_comparison}}
\resizebox{\textwidth}{!}{\begin{tabular}{|c|c|c|c|c|c|c|c|}
\hline
\multicolumn{2}{|c|}{\multirow{3}{*}{Method}}    & \multicolumn{3}{c|}{\begin{tabular}[c]{@{}c@{}}Mishra $(k=16)$\\ \cite{b11}\end{tabular}}                                          & \multicolumn{3}{c|}{Proposed work $(k=16)$}                                                                                      \\ \cline{3-8} 
\multicolumn{2}{|c|}{}                           & \multirow{2}{*}{Text file of 3920 bits} & \multicolumn{2}{c|}{\begin{tabular}[c]{@{}c@{}}Image files\\ in pixels\end{tabular}} & \multirow{2}{*}{Text file of 3920 bits} & \multicolumn{2}{c|}{\begin{tabular}[c]{@{}c@{}}Image files\\ in pixels\end{tabular}} \\ \cline{4-5} \cline{7-8} 
\multicolumn{2}{|c|}{}                           &                                         & $256\times256$                                  & $512\times512$                                  &                                         & $256\times256$                                   & $512\times512$                                  \\ \hline
\multirow{2}{*}{\begin{tabular}[c]{@{}c@{}}Information\\ density\end{tabular}} & $1^{st}$ method& \multirow{2}{*}{2.41}                   & \multirow{2}{*}{2.09}                     & \multirow{2}{*}{2.31}                    & 4.48                                    & 2.68                                      & 2.41                                     \\ \cline{2-2} \cline{6-8} 
                                     & $2^{nd}$ method&                                         &                                           &                                          & 4.39                                    & 2.63                                      & 2.37                                     \\ \hline
\end{tabular}}
\end{table}

 \normalsize

 \vspace{10pt}
\section{Results and Analysis}
\label{section4}

In this section, we show two performance results of the proposed encoding algorithm through the simulation. First, we show a reduction of the average bit error rate when one nt error occurs in the DNA sequence, compared to the randomly generated mapping table. Second, results for information density and iteration numbers are obtained when the proposed encoding algorithm is applied in raw data files. Here, we use the proposed encoding algorithm for one text file and two image files.
\subsection{The Comparison of the Average Bit Error Rate}
\label{section_simulation_mapping}
As we mentioned in Section \ref{section_greedy}, before creating a mapping table, we obtain the substitution error probability between bases by conducting experiments \cite{b15}. In this experiment, 18000 oligo sequences of length 152 nt with 300 nanogram DNA oligo pools are synthesized by Twist Bioscience. We use Illumina Miseq Reagent v3 kit (600 cycle) for sequencing and obtain 151 nt run for both forward and reverse reads. In this experiment, the substitution error probability across different bases are obtained as in Table \ref{table_error}.\par
Using these error probabilities, we create the mapping table, Table \ref{table_map}, and find the average bit error for each substitution error probability between bases. Table \ref{table_error} shows the average bit error for each case of substitution errors when one nt error occurs in the DNA sequence. Therefore, the average number of bit errors when one nt error occurs is 2.3455 bits in the proposed encoding algorithm. Every 48-ary symbol could be changed to any other 48-ary symbol except itself when one nt error occurs in the randomly generated mapping table. Therefore, the average number of bit errors can be obtained by considering all cases of converting one symbol to the other, from 0 to 1, 0 to 2, ... 47 to 48, and calculate the average bit errors caused by all cases. Then, the average number of bit errors is 2.9504 bits. In other words, the proposed mapping method could reduce the bit error by 0.6049 bits, which corresponds to $20.5\%$ reduction, compared to the randomly generating mapping method. Therefore, in this step, the proposed mapping method reduces the bit error $20.5\%$ in the DNA sequence while satisfying the maximum run-length at most three. In the constrained coding, even one nt error could cause a high bit error rate because of its high error propagation. However, the proposed algorithm prevents a single DNA error from causing a large number of bit errors.

%\begin{table}[t]
%\centering
%\caption{Information density using the second mapping method (bits/nt) \label{table_result}}
%\small
%\renewcommand{\arraystretch}{1.2}
%\begin{tabular}{|c|c|c|c|}
%\hline
%\multirow{2}{*}{$k$} & \multirow{2}{*}{Text file of 3920 bits} & \multicolumn{2}{c|}{Image files in pixels} \\ \cline{3-4} 
%                   &                            & $256\times256$          & $512\times512$         \\ \hline
%2                  & 1.8209                   & 1.8241             & 1.8242            \\ \hline
%4                  & 2.2187                   & 1.9687             & 1.8556            \\ \hline
%8                  & 3.4763                   & 2.1688             & 1.9541            \\ \hline
%16                 & 4.3923                   & 2.6340             & 2.3733            \\ \hline
%\end{tabular}
%\end{table}

\subsection{Simulation Results for Information Density and Iteration Numbers}
\label{section_simulation_result}
In our simulation, to compare with the recent work \cite{b11}, we used the same text file, and two image files as in \cite{b11}. The text file is a poem \textit{Where the Mind is Without Fear} by \textit{Rabindranath Tagore}, which consists of 3,920 bits with 490 characters. The second file is a grayscale image of an airplane with a size of $256\times256$, where each pixel consists of 8 bits. The last file is a grayscale image of \textit{'Lena'} with a size of $512\times512$, where each pixel consists of 8 bits. Each image file consists of 2,097,512 bits, 524,288 bits, respectively.\par
The experiment is held in the case of typical DNA storage, whose output has $m=3$ and $0.45\le r_{GC}\le 0.55$. In the case of source coding, any efficient scheme can be applied flexibly, but to compare with the work in \cite{b11}, the minimum variance Huffman tree code has been used. Here, we define a block of length $k$ as one symbol and we used the value $k=16$ for the comparison with \cite{b11}. Obtaining compressed data $F_{comp}$, we use the SHA-3 algorithm for randomization to balance the data. We initialize $r=0$ and the output length of $h(r)$ is 512 bits. Then we perform XOR of $F_{comp}$ and $h(r)$. Next, we use a 48-ary converting table to obtain the DNA sequence $F_{mapping}$ since we allow up to three runs. Finally, we first partitioned $F_{mapping}$ into 200 nt long DNA sequences. Then in the verification step, we check whether the final output is balanced or not, and for the balanced DNA sequences, we append the input value $r$ in the form of DNA. If the DNA sequence that does not satisfy the constraint remains, we start again from the randomization step increasing the value $r$ by 1. Repeat this process until all DNA sequences are balanced.
As mentioned in Section \ref{section_verification}, the iteration number of $I=4$ is enough for satisfying the GC balance ratio for all text and image files in our simulation. The more detailed explanation is in the last paragraph of this section.\par

%%%%%%%%%%%%%%%%%%%%%%%%

%\begin{figure}[]
%\label{fig_raw_input}
%\centering
%\subfigure[]{\includegraphics[width=0.25\linewidth]{lena.PNG}}
%\subfigure[]{\includegraphics[width=0.25\linewidth]{airplane.PNG}}
%\subfigure[]{\includegraphics[width=0.4\linewidth]{text.PNG}}
%\caption{(a) }
%\end{figure}
%%%%%%%%%%%%%%%%%%%%%%%%%%%%%%%%%%%%%%%%%%%%%%%%%%%%%%%%%%%%%%%%%%%%%
% \usepackage{multirow}
% Please add the following required packages to your document preamble:
% \usepackage{multirow}

%%%%%%%%%%%%%%%%%%%%%%%%%%%%%%%%%%%%%%%%%%%%%%%%%%%%%%%%%%%%%%%%%%%%%

For the text file, 3920 bits are compressed to 1618 bits using Huffman code ($k=16$), whose compression rate is 2.4227. Then we randomize the binary text file and convert the binary sequence to 48-ary symbols. In this step, 1618 bits are converted to 48-ary symbols of length 290 using the first mapping method. Using the second mapping method, 1618 bits are converted to 296 48-ary symbols. Then, each 48-ary symbol is converted to three nt DNA sequence using Table \ref{table_map}, which has the information density of 1.86 bits/nt. Finally, we append $r$ by converting it to the DNA sequence using Table \ref{table_map}, and then 3920 bits are converted to 875 nt DNA sequence. Table \ref{table_comparison} shows that the proposed encoding algorithm has improved the information density compared to the recent work \cite{b11}, for all files. As shown in Table \ref{table_comparison}, final information density using the first mapping method is $\frac{3920}{875}= 4.48$ bits/nt, which means one DNA nucleotide contains 4.48 bits. For the second mapping method, one DNA nucleotide contains 4.39 bits. In the mapping method, for image files of sizes $256\times256$ and $512\times512$, information density is 2.68 bits/nt and 2.41 bits/nt, for $k=16$, respectively. In the second mapping method, the information density of text file, image files of sizes $256\times256$ and $512\times512$ is 4.39 bits/nt, 2.63 bits/nt, 2.37 bits/nt, for $k=16$, respectively.\par

For another experiment, we find the number of DNA sequences of length 200 nt that satisfy the GC balance ratio between $[0.45, 0.55]$ at each iteration. According to Theorem \ref{thm_randomization} and Table \ref{table_randomization}, the GC balance ratio between $[0.45, 0.55]$ can be satisfied within four times of iterations. Table \ref{table_iteration} shows the experimental results for Theorem \ref{thm_randomization} and Table \ref{table_randomization}. As shown in Table \ref{table_iteration}, All DNA sequences satisfy the desired constraint within four times of iterations, and therefore only one redundancy is appended for each DNA sequence.

%%%%%%%%%%%%%%%%%%%%%%%%%%%%%%%%%%%%%%%%%%%%%%%%%%%%%%

\begin{table}[t]
\centering
\small
\renewcommand{\arraystretch}{1.2}
\caption{The number of DNA sequences with $n=200$ that satisfy the GC balance ratio between $[0.45, 0.55]$ at each iteration \label{table_iteration}}
\begin{tabular}{|c|c|c|c|}
\hline
\multirow{2}{*}{$I$} & \multirow{2}{*}{Text file of 3920 bits} & \multicolumn{2}{c|}{Image files in pixels} \\ \cline{3-4} 
                       &                            & $256\times256$          & $512\times512$         \\ \hline
1                  & 4                          & 899                    & 3921                   \\ \hline
2                   & 0                          & 85                    & 429                   \\ \hline
3                  & 0                          & 6                    & 42                   \\ \hline
4                  & 0                          & 0                    & 4                   \\ \hline
success rate                  & $100\%$                          & $100\%$                    & $100\%$                   \\ \hline

\end{tabular}
\end{table}

\vspace{10pt}
\section{Conclusion}
\label{section5}
In the proposed DNA encoding algorithm, we applied the greedy algorithm in the mapping step to reduce the average bit error when one nt error occurs in the DNA sequence. As a result, there is $20.5\%$ of the average bit error reduction compared to the randomly generated mapping method. In addition to this feature, the proposed encoding algorithm is robust to error propagation. Also, we used the iterative encoding algorithm with three steps to convert the raw file into the DNA strand. Compared to the existing works, the proposed encoding algorithm has high information density and flexible feature for the desired GC balance ratio of $0.5\pm\alpha$ and the maximum run-length of $m$. Therefore, the proposed encoding algorithm could reduce the synthesizing cost and error that could be occurred in the DNA storage processes. Since it is important to retrieve the data without any error in the storage system, the proposed encoding algorithm would be very useful in the field of the DNA storage because of its robustness of error and flexible property.

\vspace{10pt}
\bibliographystyle{IEEE}

% biography section
% 
% If you have an EPS/PDF photo (graphicx package needed) extra braces are
% needed around the contents of the optional argument to biography to prevent
% the LaTeX parser from getting confused when it sees the complicated
% \includegraphics command within an optional argument. (You could create
% your own custom macro containing the \includegraphics command to make things
% simpler here.)
%\begin{IEEEbiography}[{\includegraphics[width=1in,height=1.25in,clip,keepaspectratio]{mshell}}]{Michael Shell}
% or if you just want to reserve a space for a photo:

% You can push biographies down or up by placing
% a \vfill before or after them. The appropriate
% use of \vfill depends on what kind of text is
% on the last page and whether or not the columns
% are being equalized.

%\vfill

% Can be used to pull up biographies so that the bottom of the last one
% is flush with the other column.
%\enlargethispage{-5in}

% that's all folks
\end{document}